\documentclass{vldb}
\usepackage{cleveref}
\usepackage{graphicx}
\usepackage{balance}  
\usepackage{algorithm}
\pdfoutput=1
\usepackage{algpseudocode}
\usepackage{amsmath}
\usepackage{subfigure}

\vldbTitle{Efficient Suspected Infected Crowds Detection Based on Spatio-Temporal Trajectories}
\vldbAuthors{Huajun He et al}
\vldbDOI{https://doi.org/10.14778/xxxxxxx.xxxxxxx}
\vldbVolume{12}
\vldbNumber{xxx}
\vldbYear{2019}

\begin{document}
	
	
	\title{Efficient Suspected Infected Crowds Detection Based on Spatio-Temporal Trajectories}

	
	
	%
	%
	%
	%

	\numberofauthors{6}
	\author{
		%
		%
		\alignauthor
		Huajun He\\
		\affaddr{Southwest Jiaotong University}\\
		\affaddr{Chengdu, China}\\
		\email{hehuajun@my.swjtu.edu.cn}
		\alignauthor
		Ruiyuan Li\\
		\affaddr{Xidian University}\\
		\affaddr{Xi'an, China}\\
		\email{ruiyuan.li@jd.com}
		\alignauthor
		Rubin Wang\\
		\affaddr{Southwest Jiaotong University}\\
		\affaddr{Chengdu, China}\\
		\email{wangrubin3@jd.com}
		\and
		\alignauthor 		Jie Bao\\
		\affaddr{JD Intelligent Cities Research}\\
		\affaddr{Beijing, China}\\
		\email{baojie3@jd.com}
		\alignauthor 		Yu Zheng\\
		\affaddr{JD Intelligent Cities Research}\\
		\affaddr{Beijing, China}\\
		\email{zheng.yu@jd.com}
		\alignauthor 		Tianrui Li\\
		\affaddr{Southwest Jiaotong University}\\
		\affaddr{Chengdu, China}\\
		\email{trli@swjtu.edu.cn}
			}
%


	\maketitle
	
\begin{abstract}
	Virus transmission from person to person is an emergency event facing the global public. Early detection and isolation of potentially susceptible crowds can effectively control the epidemic of its disease. Existing metrics can not correctly address the infected rate on trajectories. To solve this problem, we propose a novel spatio-temporal infected rate (IR) measure based on human moving trajectories that can adequately describe the risk of being infected by a given query trajectory of a patient. Then, we manage source data through an efficient spatio-temporal index to make our system more scalable, and can quickly query susceptible crowds from massive trajectories. Besides, we design several pruning strategies that can effectively reduce calculations. Further, we design a spatial first time (SFT) index, which enables us to quickly query multiple trajectories without much I/O consumption and data redundancy. The performance of the solutions is demonstrated in experiments based on real and synthetic trajectory datasets that have shown the effectiveness and efficiency of our solutions.
\end{abstract}

\section{Introduction}

Human to human virus-borne infections has been a public health concern. A new coronavirus was named "SARS-CoV-2," and the disease it caused was called "Coronavirus Disease 2019" (abbreviated as "COVID-19")~\cite{zhao2020comparative}. COVID-19 pandemic is an urgent emergency facing the world. In the absence of a vaccine, early detection, early reporting, early isolation, and early treatment~\footnote{http://dwz1.cc/mAz79Dq} have proven to be the most effective measures to prevent the spread of the epidemic. 

With the rapid development of mobile internet and locate service, massive spatio-temporal data have been generating from applications. Human activity trajectory is a typical spatio-temporal data, including longitude, latitude, and time. Given a trajectory $Q$ of the confirmed patient, Suspected Infected Crowds Detection (SICD) aims to detect close contacts through the spatio-temporal correlation of trajectories. As shown in Figure~\ref{fig:sic}, we search the ordinary people who have occurred within the spatio-temporal range which can be infected with the location where the patient has appeared and then determine the probability of infection rate based on their contact distance and duration. SICD helps local governments to investigate suspected people and find close contacts, isolate and protect them in time to prevent further spread of the epidemic. For improving the accuracy of SICD, it is necessary to consider the spatio-temporal correlation in each location of Q to describe the infection rate. The data volume of the underlying trajectory database used for SICD is enormous. To avoid massive memory consumption, we leverage a spatio-temporal index to manage trajectories in the NoSQL database.

In the epidemiological analysis, analyzing the relationship between people in spatial and temporal is a very standard and important analytical method. By looking at the spatial and temporal relationship, we can draw accurate close contact conclusions. In the 19th century, Snow~\cite{vinten2003cholera}, studied spatio-temporal data such as maps and found that the source of pollution in cholera cases was not air, but from public pumps on Broad Street and transmitted through contaminated drinking water. At his appeal, authorities closed and diverted pump valves to control cholera. The successful prevention of cholera is directly related to the result of spatio-temporal data analysis and is the most classic example of spatio-temporal big data analysis.

We can acquire the multiple spatio-temporal data information related to the risk of infection within the given spatio-temporal ranges through the trajectory of human activity. Moving trajectory is typical spatio-temporal data. By investigating the patient's moving trajectory, we can know who the patient is in close contact. In response to this demand, Tang discovered object groups that travel together from streaming trajectories~\cite{tang2012discovery}. In his study, in order to find travel companions, the system needs to cluster the objects of each snapshot of the query trajectory and intersect the clustering results, and retrieve the objects that move together. It is an efficient system with high precision, and it can discover crowds with a long companion. However, it can not find crowds whose local companion is long in one snapshot but total companion is short in all snapshots. Besides, many scholars have done much research on the spatial and temporal similarity between trajectories. Some of these methods mainly focus on how to extend existing trajectory similar search algorithms (e.g., ED, DTW, and Fréchet~\cite{toohey2015trajectory}). They have excellent performance for detecting duplicated or redundant trajectories in the database. However, to ensure accuracy, they require the consistency of the sampling rate of the two trajectories to be relatively high, and they lack the anti-noise ability. Some other trajectory similarity algorithms study the spatio-temporal correlation. They calculate spatial and temporal correlations separately and then combine them into overall similarity.  A liner combination method (e.g.,~\cite{shang2014personalized}) combines the spatial and temporal into a spatio-temporal similarity metric. Other existing trajectory similarities (e.g.,~\cite{bakalov2005efficient,bakalov2006continuous}) use a time interval to limit the resemblance of two trajectories. Unfortunately, they calculate similarity on the entire trajectory. However, in the actual data set, the human trajectory is not always successive and may has a larger spatial and temporal range. Therefore, in many applications, the trajectory must be segmented. However, the local segment similarity of the trajectory is not comprehensive in existing metrics. Hence, in this paper, we propose a measure that weights every segment of the trajectory, because the longer a segment of the query trajectory stays, the more possible that others will meet it. In each segment, we consider not only spatial nearing but also temporal proximity in every location, which makes sure the spatio-temporal closed location has a high infected rate.
\begin{figure}
	\centering
	\includegraphics[width=3.0in]{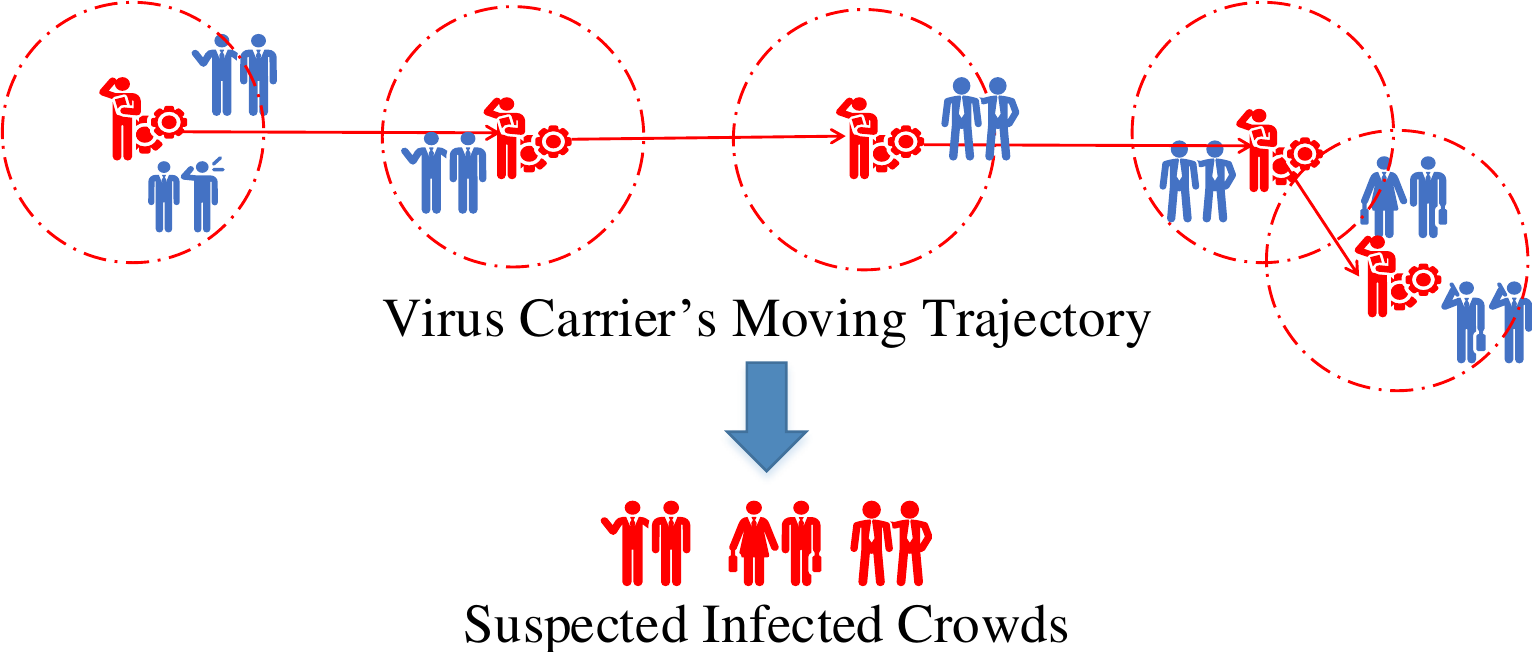}
	\caption{An example of Suspected Infected Crowds Detection.}
	\label{fig:sic}
\end{figure}
COVID-19~\cite{zhao2020comparative} is an urgent emergency facing the world. It has an incubation period. Thus, people infected with COVID-19 do not immediately have severe symptoms after they infect the human body, with an average of 5 to 6 days and a range of 1 to 14 days. Coronavirus infection during the incubation period is also infectious. Therefore, it is necessary to confirm all GPS records of the confirmed diagnosis from the prior incubation period, which is a large amount of data with a large spatio-temporal range. However, most of the existing solutions should load all data into memory, which limits the scalability. In this paper, we first divide the long and large trajectory into several sub-segments with suitable length and spatio-temporal range. Then, we build an XZ2T~\cite{li2020just} index to manage large segments in the NoSQL database via an efficient platform JUST~\cite{li2020just}, which guarantees the scalability of our solutions. Many methods mainly focus entirely on similarity. Thus, they are slow due to the large consumption for calculating big trajectory similarity. Therefore, recent researches (e.g.,~\cite{shang2014discovery,shang2017trajectory,shang2014personalized,ta2017signature}) have focused on some pruning strategies to save time. They build indexes on trajectories to avoid computing all trajectories similarity, which greatly accelerates query speed. In this paper, we only search the smaller spatio-temporal range of each segment to reduce the range of the candidate set and propose some pruning strategies, which help reduce the computation. In most scenarios, it is always to query close contacts for multiple patients. Thus, we build an efficient SFT index to reallocate segments with the similar spatio-temporal range together, which effectively reduces the I/O cost and data redundancy.

Using our algorithm, we helped Beijing find more than 500 high-risk close contacts within 20 days. Until March 1, we assisted Suqian in discovering a quarter of newly diagnosed patients with COVID-19 in the city. Within China, 18 provinces and cities such as Guangzhou, Nanjing, and Chengdu used this algorithm as part of the high-risk population analysis system.

To sum up, the contributions of this paper are as follows.
\begin{itemize}
	\item We propose a new infection rate (IR) metric that takes into account both the spatial and temporal proximity in all segments of trajectory and is suitable for the Suspected Infected Crowds Detection (SICD). It can also use to recognize similar trajectories, detect close contacts, mine companion, and monitor high-risk groups.
	\item We store primary trajectories in the NoSQL database and only need to search a small spatio-temporal range data when it comes to a query trajectory, which reduces the memory consumption and guarantees the scalability of our solutions.
	\item We leverage some effective pruning strategies to avoid many invalid calculations. 
	\item We develop an SFT index to reduce I/O communication and data redundancy. 
	\item We conduct extensive experiments on trajectory sets to validate the performance of the proposed algorithms.
\end{itemize}

The rest of the paper is organized as follows. Section 2 introduces the basic definitions and trajectory infected rate. The framework of our solution is presented in Section 3. The trajectory infected rate query is described in Section 4, while the trajectory infected rate join query for multiple patients is in Section 5. The experimental results are presented in Section 6. Related work is illustrated in Section 7, and conclusions and future works are shown in Section 8.

\section{PRELIMINARIES}
In this section, we introduce the basic definitions and spatio-temporal operations in our present approaches.

\subsection{Trajectory}
The trajectory is a typical representation of a set of spatio-temporal locations for the same user. A location in the trajectory is of the form (longitude, latitude, time), and all locations of trajectory are sorted by its timestamp. A trajectory is defined as follows.
\newdef{definition}{Definition}
\begin{definition}
	A trajectory $T$ of moving object is a time-ordered locations $<l_1,l_2,...,l_n>$, where location $l_i = (p_i,t_i)$, $ i \in [1,n]$, with $p_i$ is a spatial point, $t_i$ is a timestamp, and $n$ is the location size.
\end{definition}
\subsection{Segmentation}
In order to fully depict the virus carrier's infected trajectory, we must collect all of its spatio-temporal records generated by GPS terminals from the incubation period to isolation. However, records may not be collected continuously, such as when the terminal is shut down for a while. Thus, the spatial or temporal interval between the two nearest records may be extremely large. If the entire trajectory is stored as a whole, a sizeable spatio-temporal range is required to contain this trajectory, and collected records may be intermittent that two nearest GPS records cannot be directly connected. Meanwhile, the spatio-temporal nearest locations always have similar characteristics. Storage together can improve the efficiency of the infected rate calculation because data redundancy can be avoided. Therefore, the trajectory must be segmented. In this paper, we use the stay point detection algorithm~\cite{cloudtp} to segment the trajectory. As shown in Figure~\ref{fig:segment}, we divide trajectory into four segments marked with red boxes, where the spatial distance and time interval between any two locations in any segment do not exceed fixed thresholds (e.g., 200m and 30minutes in Figure~\ref{fig:segment}), respectively.
\begin{figure}
	\centering
	\includegraphics{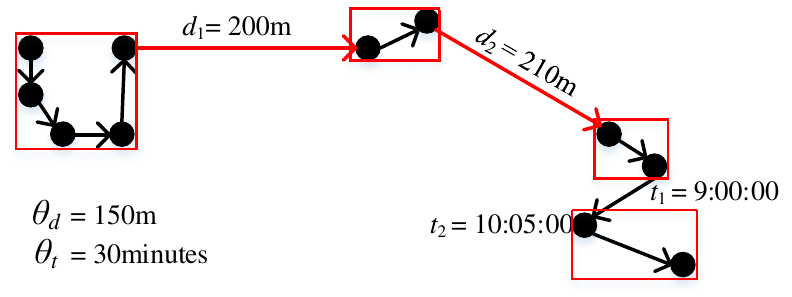}
	\caption{A sample of segmentation.}
	\label{fig:segment}
\end{figure}

\begin{definition}
	A trajectory  $T$ can be represented by segments. $ T = <s_1,s_2,...,s_m>$, where $m$ is the segment size of this trajectory, segment $s_i = <l_{i1},l_{i2},...,l_{ik}>$ is a sub-trajectory, and segments are sorted by the timestamps of their start locations.
\end{definition}

\subsection{Spatiotemporal Operations}
In this paper, we focus on spatio-temporal operations over virus carrier's moving trajectories to link their close contact users. Thus, in this part, we discuss spatio-temporal operations of our propose approaches.

Carriers of COVID-19 affect people with whom they have close contact. Therefore, the factors of infected rate are related to the spatio-temporal distance between the ordinary person and the patient. Thus, each position where the carrier appears has an influential spatio-temporal range.  
\begin{definition} 
	Given a location $l$, a spatial infected range threshold $\theta_d$ and a temporal infected range threshold $\theta_t$. $STR(l,\theta_d,\theta_t)$ represents an influential spatio-temporal range of location $l$. Formally,
	$$STR(l,\theta_d,\theta_t) = \{R|\forall r \in R(|r.t - l.t| \le \theta_t \land dist(r.p,l.p) \le \theta_d)\}$$
\end{definition}

In the range $STR(l,\theta_d,\theta_t)$ of location $l$, we calculate spatio-temporal correlation of $l$ for trajectory.

\begin{definition} Given a location $l$, a trajectory $T$, a distance threshold $\theta_d$ and a time threshold $\theta_t$. The spatio-temporal correlation is defined as follows:
	\begin{equation}
	st\_cor(l,T) = \max_{v \in T \land v \in STR(l,\theta_d,\theta_t)}{st\_dist(l,v)}
	\label{st_cor},
	\end{equation}
\end{definition}
where $st\_dist(l,v)$ is the spatio-temporal correlation between $l$ and a location $v$ of trajectory $T$. Formally,
\begin{equation}
st\_dist(l,v) = \lambda e^{-\frac{dist(l.p,v.p)}{\theta_d}} + (1-\lambda)e^{-\frac{|l.t - v.t|}{\theta_t}}
\label{st_cor_s},
\end{equation}
where parameter $\lambda \in [0,1]$ controls the relative importance of the spatial and temporal correlation. $\frac{dist(l.p,v.p)}{\theta_d}$ and $\frac{|l.t - v.t|}{\theta_t}$ normalize the effects of spatial distance and time interval to the same range. Note, while trajectory $T$ does not intersect with $STR(l,\theta_d,\theta_t)$, the spatiotemporal correlation $st\_cor(l, T)$ is 0.

\subsection{Trajectory Infected Rate}
\subsubsection{Segment Infected Rate}
Given a segment $s$ of virus carrier's trajectory and a trajectory $T$, the infected rate between $s$ and $T$ is defined as follows:
\begin{equation}
IR(s, T) = \frac{\sum_{l \in s}{st\_cor(l,T)}}{|s|}
\label{IR_s},
\end{equation}
where $l$ is a location of $s$ and $|s|$ represents the number of locations that $s$ owns.
\subsubsection{Trajectory Infected Rate}
Given a virus carrier's query trajectory $Q$ and a trajectory $T$, the infected rate between $Q$ and $T$ is defined as follows:
\begin{equation}
IR(Q, T) = \sum_{i=1,s_i \in Q}^{m}{P(s_i) * IR(s_i,T)}
\label{IR_T},
\end{equation}
where $s_i$ is a segment of $Q$ and $P(s_i)$ represents the potential infected probability of each segment in $Q$. The probability $P(s_i)$ is determined by time span in segment, on account of more time the carrier stays more risk of infection the others may have. Therefore, $P(s_i)$ is defined as follows:
\begin{equation}
P(s_i) = \frac{s_i.et - s_i.st + 1}{\sum_{j=1}^{m}{(s_j.et - s_j.st + 1)}},
\label{ps}
\end{equation}
where $m$ represents the number of segments in the query trajectory $Q$, $s_i.st$ represents the start time of segment and $s_i.et$ represents the end time of segment, respectively.   

\subsection{Problem Definitions}
Given a query trajectory $Q$, a set of trajectories $\mathbf{T}$ and a threshold $\theta$, the trajectory infected rate query(IRQ) finds a set of trajectories $\mathbf{T}'$ from the set whose trajectory infected rate exceeds $\theta$, i.e., $\forall T \in \mathbf{T}'(IR(Q,T) > \theta)$.

Given a set of query trajectories $\mathbf{Q}$, a set of trajectories $\mathbf{T}$ and a threshold $\theta$, the trajectory infected rate join query(IRJQ) finds a set of all trajectory pairs from the two sets whose trajectory infected rate exceeds $\theta$, i.e., $\forall (Q_i,T_j) \in \mathbf{Q \Join T}(IR(Q_i,T_j) > \theta)$.

\section{Framework}

Figure~\ref{fig:framework} depicts the architecture of the trajectory infected rate query, which consists of three processes: Data Preprocessing, Indexing and Storing, and Infectivity Query. 
\begin{figure}
	\centering
	\includegraphics[width=3in]{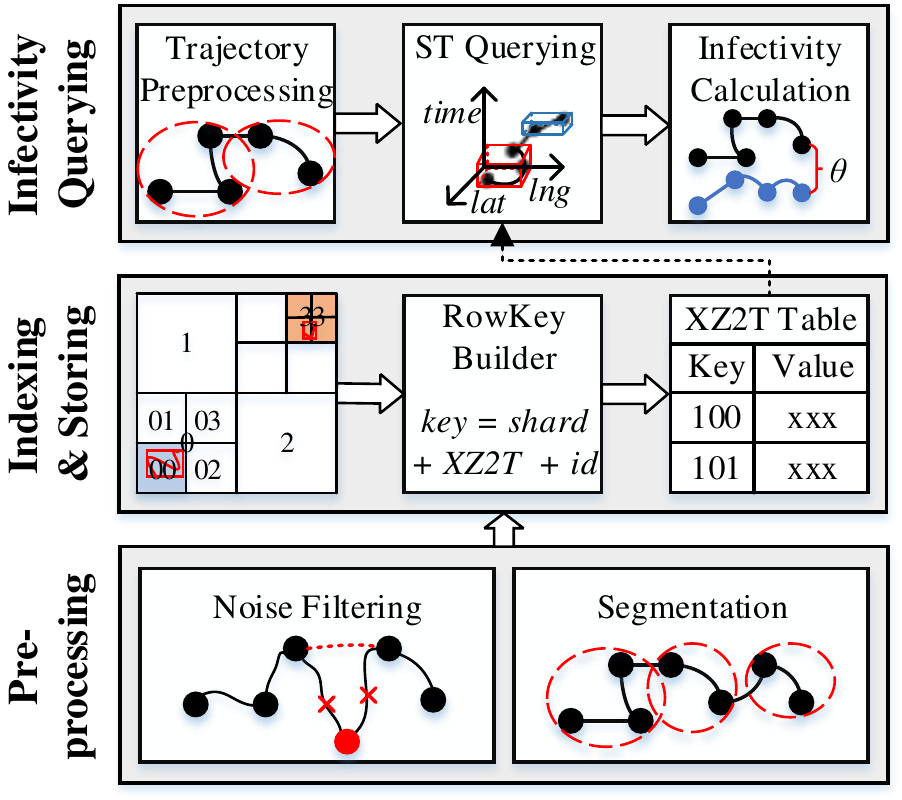}
	\caption{The framework of our solutions.}
	\label{fig:framework}
\end{figure}

\subsection{Data Preprocessing}
In many applications, trajectory preprocessing is not only necessary for filtering noise but also crucial for indexing and storing. As depicted in the bottom-most box of Figure~\ref{fig:framework}, the process of preprocessing contains two main tasks: 1) noise filtering, which eliminates outlier GPS records that may be caused by the weak signal of GPS terminals; 2) segmentation, which breaks a long trajectory into suitable short trajectories. This paper mainly focuses on the trajectory indexing and infectivity query. For more details about trajectory preprocessing, please refer to our previous work~\cite{cloudtp}.

\subsection{Indexing and Storing}
As shown in the middle box of Figure~\ref{fig:framework}, we use the XZ2T index to organize the segment and then store it as a table into the NoSQL database via JUST~\cite{li2020just}, which can efficiently and conveniently manage big spatio-temporal data. 

Spatio-temporal range query is a necessary step in our algorithm. Indexing is essential for the processing of spatio-temporal query. Thus, we build an XZ2T index on the segment to effectively support the spatio-temporal range query. 

XZ2T index is an extension of the XZ2 index~\cite{boxhm1999xz}, which projects a geographical polygon with a time range onto a one-dimensional value. XZ2 index is based on XZ-Ordering, a Space-Filling Curve for Spatially Extended Objects. It uses a sophisticated coding scheme for a polygon, which maps the polygon into the integer domain. As shown in Figure~\ref{fig:xz2t} (a), XZ2 index divides the root element into four sub-elements with equal size, which are numbered from 0 to 3. Then, the XZ2 index recursively numbers each sub-space until the maximum resolution is reached. Finally, we can get a sequence formed by successively traversing numbers. A polygon is represented by the most appropriate element or xelement of the xz2 index, which can completely cover the polygon. The xelement is an enlarged area of the element in xz2 index (i.e., the xelement of ``210" represents the area covered by the element ``21", and the width and height of $t_2$ are lower than ``210". Thus, instead of ``21", we can use ``210" represents $t_2$. Similarly, ``12" represents $t_1$, ``221" represents $t_3$, ``0" represents $t_4$, respectively). However, XZ2 index only supports spatial data. Therefore, considering the time dimension, the XZ2T index is designed, which allocates each disjoint period an XZ2 index, as shown in Figure~\ref{fig:xz2t} (b). Specifically, given a segment's spatio-temporal range $(mbr,st, et)$, we first calculate the period number $bin$ of $st$ according Equ (\ref{bin}), then calculate its XZ2 index number by using XZ-ordering function $XZ2(mbr)$. Finally, we combine the period number and XZ2 index to indicate the spatio-temporal range of segment. Note, our segmentation algorithm guarantees the span of $et-st$ not greater than $periodLen$ and $et$ belongs to the same $bin$ of $st$.  
\begin{equation}
bin = Unit.between(t, epoch) / periodLen
\label{bin}
\end{equation}

In Equ (\ref{bin}), $epoch$ is the reference time (e.g., 1970-01-01T00:00:00), and $periodLen$ is the span of a period. $Unit$ represents the unit of time (e.g., day, week, month, year).

To support offline infectivity query and avoid all trajectories is persisted in memory because memory resources are expensive and insufficient, we store segments of trajectory to NoSQL database (e.g., HBase) via JUST. The key of our table is consisted by a shard, XZ2T index, and sid:
$$key = shard + XZ2T(s) + sid,$$
where $shard$ is a hash number to achieve load balance; $XZ2T$ encodes the segment's spatio-temporal information; $sid$ is the unique id of each segment. 
\begin{figure}
	\centering
	\includegraphics[width=3.3in]{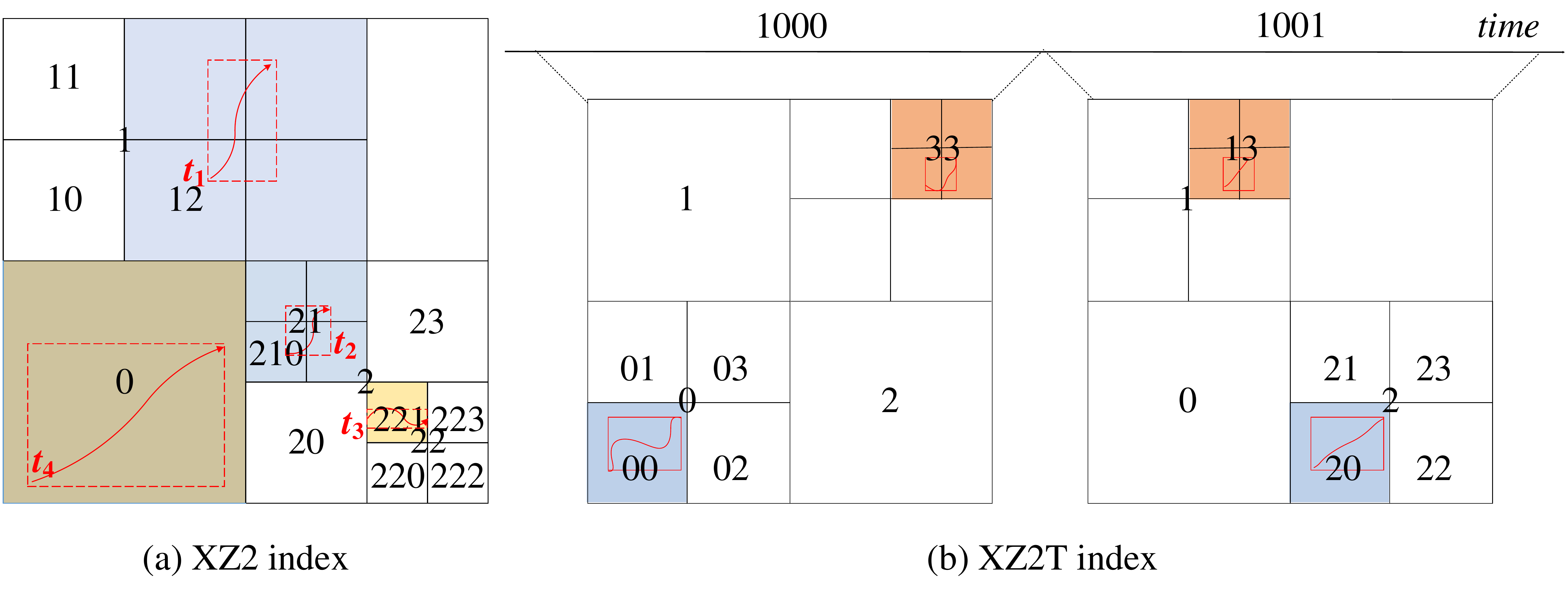}
	\caption{The example of XZ2 and XZ2T index.}
	\label{fig:xz2t}
\end{figure}

\subsection{Infectivity Querying}
The top-most box of Figure~\ref{fig:framework} shows two tasks of infectivity query, including Infected Rate Query(IRQ) and Infected Rate Join Query(IRJQ). IRQ finds trajectories that had close contact with a virus carrier (see Section 4). IRJQ finds trajectory pairs from two sets whose infectivity exceeds $\theta$ (see Section 5).  In the following sections, we introduce the details of IRQ and IRJQ. 
	
\section{Infection Rate Query}
\subsection{Main idea}
We propose a spatio-temporal correlation-based infectivity query. First, we break the long query trajectory $Q$ to suitable segments $S$. Second, we extract and extend the spatio-temporal range of each segment. Third, we query all infected segments covered by the spatio-temporal ranges of $S$ from database and aggregate the same trajectory's segments together (Section 4.2). Furthermore, we prune the trajectories whose infected rate can not greater than $\theta$. After pruning, we calculate the infectivity between the remain trajectories and $Q$. Finally, we filter the trajectories whose infected rate is lower than $\theta$ out and return the results (Section 4.3).
\subsection{ST Query}
\subsubsection{Main idea}
\begin{figure}
	\centering
	\includegraphics[width=3in]{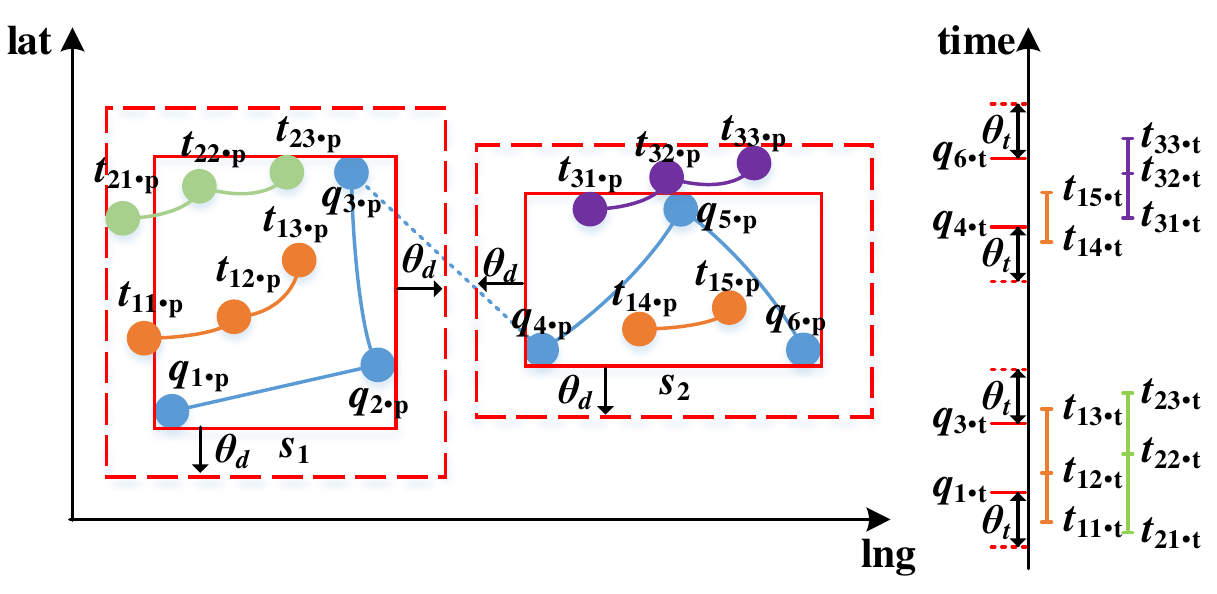}
	\caption{A sample of spatio-temporal query.}
	\label{fig:base_query}
\end{figure}
In this solution, we must query segments from database with XZ2T-Index around each location of the query trajectory. However, there are many locations need to access database that cause the massive queries and data redundancy. Thus, we cluster close locations to a spatio-temporal range and use this range to query the data once from the database. 

The raw long trajectory may need a large spatial and temporal range to cover it entirely. For example, a trajectory may go to many places for ten days. Hence, we break the trajectory where the spatial distance and time interval of any two nearest locations exceeds the fixed thresholds. As shown in Figure~\ref{fig:base_query}, the query trajectory $Q$ (the trajectory with blue in  Figure~\ref{fig:base_query}) is divided into two segments ($s_1$ and $s_2$). Each segment has a $mbr$ (minimum bounding rectangle, e.g., the red solid line rectangle of Figure~\ref{fig:base_query}) that covers all locations and a time range that starts from the first location's time and ends of the last location's time (e.g., the solid red lines between $q_1.t$ and $q_3.t$ in Figure~\ref{fig:base_query}). Then, we extend the $mbr$ outward $\theta_d$ (e.g., the red dotted line rectangle in Figure~\ref{fig:base_query}) and enlarger the time range of $\theta_t$ (e.g., the red dotted lines near $q_1.t$ and $q_3.t$ in Figure~\ref{fig:base_query}), where $\theta_d$ and $\theta_t$ control the infected spatio-temporal range. Later we query segments from the database via extended spatio-temporal ranges. As shown in Figure~\ref{fig:base_query}, the query trajectory $Q$ is divided into two segments $s_1$ and $s_2$. Segments of trajectory $t_1$ and $t_2$ are queried by $s_1$ and  $t_1$ and $t_3$ are queried by $s_2$, respectively. We group segments of the same trajectory together. Finally, $t_1,t_2$ and $t_3$ are the candidate trajectories of the query trajectory $Q$. 

\subsubsection{Query from XZ2TTable}
We store data in the NoSQL database. Each segment is saved as the form of ``(key, value)" with file dictionary sort index, and the key consists of XZ2T index and other information. Thus, generating accurate and smaller key scan ranges for query processing can significantly reduce I/O costs. 

\textbf{Key Scan Ranges Generation.} 
First, we give a spatio-temporal range of the query segment, which represents as a time range and a spatial range. Second, we extract the period numbers overlapped by the time range. Third, we generate the spatial scan ranges calculated by the XZ2 index~\cite{boxhm1999xz}. Later, for each period number, we execute the scans. Finally, we refine the result to make sure that exactly in the spatio-temporal field. 

The spatial scan range is generated as follows:  (1) Starting the recursive access from the root node by the breadth-first search; (2) if the current node intersects the spatial query window partly, the index value of this node is added to the scan queue and  recursive access child nodes until arriving the max resolution; (3) if the query window completely covers the current node, we put the index range represented by it and all its child nodes into the scan queue; (4) when the leaf nodes intersect, the index values of the leaf nodes are put into the queue; (5) if there is no intersecting node, the node is skipped directly. Finally, we combine the consecutive values in the scan queue to form the final scan range.

\subsection{Pruning}
The calculation of IRQ is time-consuming. Therefore, we develop several pruning strategies to avoid unnecessary calculations.
\newtheorem{theorem}{Lemma}
\begin{theorem}
	Let $s$ represents a segment of query trajectory $Q$ whose extended spatio-temporal range intersects with the candidate trajectory $T$. $T$ must satisfy Equ~(\ref{theorem1}):
	\begin{equation}\sum_{s \in Q\cap T }P(s) \ge \theta
	\label{theorem1}
	\end{equation}
	\label{Lemma1}
\end{theorem}

\begin{proof}
	Clearly $IR(s,T) \in [0,1]$. If $T$ does not intersect with $s$, $IR(s,T)$ equals 0, and the maximum value of $IR(s,T)$ is 1. Thus, based on Equ~\ref{IR_T}, we have that
	
	\begin{math}
	\begin{aligned}
	IR(Q, T) &= \sum_{s \in Q}{P(s) * IR(s,T)} \\
	& \leq  \sum_{s \in Q\cap T}P(s)* IR(s,T) + \sum_{s \in Q \not\cap T }P(s)* IR(s,T)\\
	& \leq \sum_{s \in Q\cap T}P(s)
	\end{aligned}
	\end{math}
	
	Thus, if $\sum_{s \in Q\cap T}P(s) < \theta$, the $IR(Q,T) < \theta$. Therefore, $\sum_{s \in Q\cap T}P(s)$ must be equal or greater than $\theta$.
	\label{proof1}
\end{proof}

\begin{theorem} The $IR(s, T)$ must satisfy Equ (\ref{l2}):
	\begin{equation}IR(s,T) \ge \frac{\theta -1 + P(s)}{P(s)}
	\label{l2}
	\end{equation}
	\label{theorem2}
\end{theorem}

\begin{proof} 
	Let $s_i$ represents a segment of $Q$. Based on Equ~\ref{ps}, we have $\sum_{i=1}^{m}P(s_i) =1$. Then by Equ~\ref{IR_T}, we have
	
	\begin{math}
	\begin{aligned}
	IR(Q, T) &= \sum_{j=1}^{m}{P(s_j) * IR(s_j,T)} \\
	& \leq P(s_i)*IR(s_i,T) + \sum_{j=1, j \neq i}^{m} {P(s_j)}  \\
	& \leq P(s_i)*IR(s_i,T) + 1 - P(s_i) 
	\end{aligned}
	\end{math}
	
	Thus, if $IR(s_i, T) < \frac{\theta -1 + P(s_i)}{P(s_i)}$, then $IR(Q,T) < \theta$, and thus trajectory $T$ can be entirely pruned.
	\label{proof2}
\end{proof}

\begin{theorem} $IR(s,T)$ must satisfy Equ (\ref{lm3}):
	\begin{equation}IR(s,T) \ge \frac{\theta -\sum_{qs \in Q\cap T,qs \neq s}{P(qs)}}{P(s)}, 
	\label{lm3}
	\end{equation}
	\label{theorem3}where $qs$ is the segment of $Q$.
\end{theorem}
\begin{proof} 
	By combining Lemmas~\ref{Lemma1} and ~\ref{theorem2}, we have
	
	\begin{math}
	\begin{aligned}
	IR(Q, T) &= \sum_{s_j \in Q}{P(s_j) * IR(s_j,T)} \\
	& \leq  P(s_i)*IR(s_i,T) + \sum_{s_j \in Q\cap T, i \neq j} {P(s_j)},
	\end{aligned}
	\end{math}
	
	where $s_i \in Q$. Thus, if $IR(s_i, T) < \frac{\theta -\sum_{s_j \in Q\cap T, i \neq j}{P(s_j)}}{P(si)}$, then $IR(Q,T) < \theta$, namely, the trajectory $T$ can be entirely pruned.
	\label{proof3}
\end{proof}

\begin{theorem} Let $s_i$ represent a segment which intersects with the candidate trajectory $T$. Then
	
	\begin{equation}IR(s_i,T) \ge \frac{\theta -\sum_{j = 1}^ {i-1}{IR(s_j,T) * P(s_j)} - \sum_{j=i+1,s_j \in Q \cap T}^{m}P(s_j)}{P(s)} 
	\end{equation}
	\label{theorem4}
	
\end{theorem}

\begin{proof} 
	By combining Lemmas~\ref{Lemma1},~\ref{theorem2} and ~\ref{theorem3}, we have 
	
	\begin{math}
	\begin{aligned}
	IR(Q, T) &= \sum_{i=1,s_i \in Q}^{m}{P(s_i) * IR(s_i,T)} \\
	& = IR(s_i,T)* P(s_i) + \sum_{j=1}^{i-1} {P(s_j)*IR(s_j,T)}  \\
	& + \sum_{j=i+1}^{m} {P(s_j)*IR(s_j,T)} \\
	& \le  IR(s_i,T)* P(s_i) + \sum_{j=1}^{i-1} {P(s_j)*IR(s_j,T)} \\
	& + \sum_{j=i+1,s_j \in Q \cap T}^{m} {P(s_j)} 
	\end{aligned}
	\end{math}
	
	Thus, if $IR(s_i,T) < \frac{\theta -\sum_{j = 1}^ {i-1}{IR(s_j,T) * P(s_j)} - \sum_{j=i+1,s_j \in Q \cap T}^{m}P(s_j)}{P(s_i)}$,  then $IR(Q,T) < \theta$. Therefore, trajectory $T$ can be entirely pruned.
	\label{proof4}
\end{proof}
\begin{figure}
	\centering
	\includegraphics[width=3.3in]{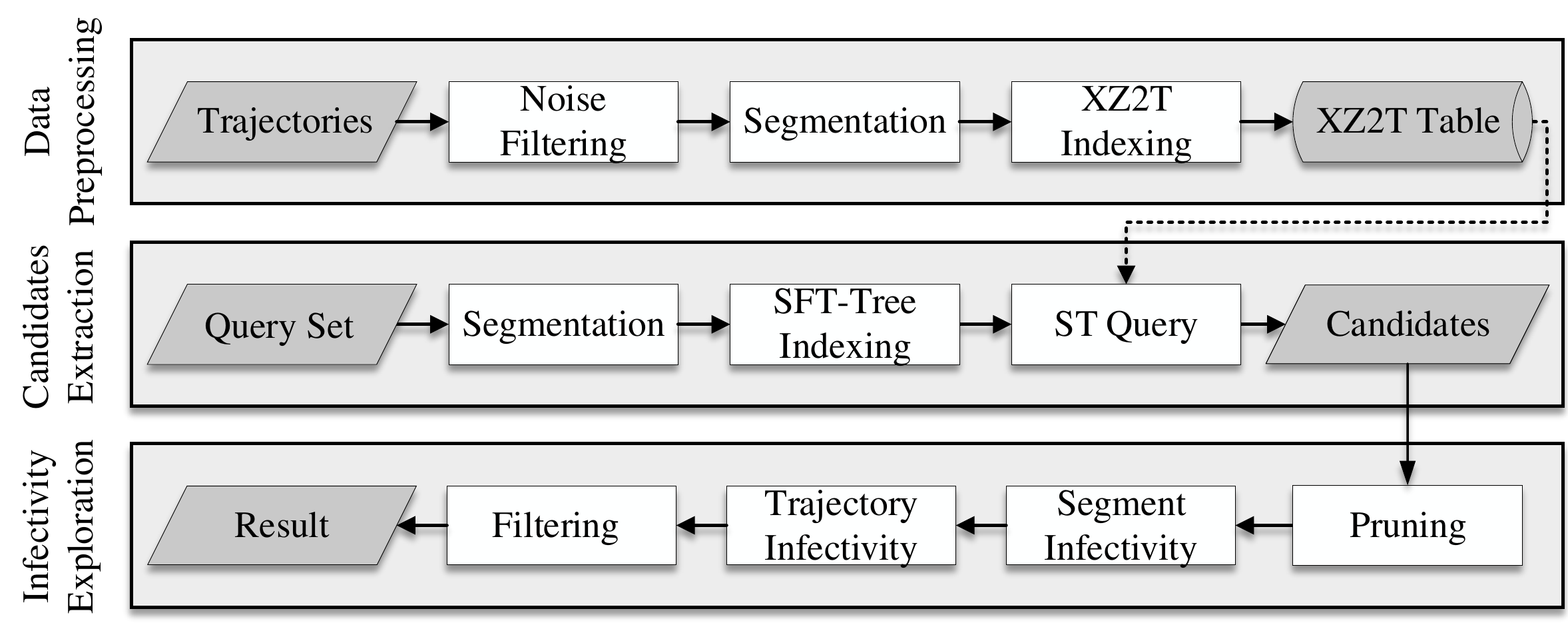}
	\caption{The structure of IRJQ.}
	\label{fig:p_fm}
\end{figure}
\begin{algorithm}[htb]
	\caption{ Infected Rate Query.}
	\label{alg:irq}
	\begin{algorithmic}[1]
		\Require
		The query trajectory, $Q$;
		The candidate trajectories, $\mathbf{T}$;
		A threshold, $\theta$;
		\Ensure 
		Trajectories with infected rate exceed $\theta$;
		\State $result = new ~ ArrayList()$; $S = segmentation(Q);$
		\For{each $T_i\in \mathbf{T}$}
		\State //Lemma 1
		\State	$sum = 0$;
		\For{each $s$ \ in $S \cap T_i$}
		\State	$sum = sum + P(s)$;
		\EndFor
		\If{$sum < \theta$}
		continue;
		\EndIf
		\State $totalIR = 0$;  $remPS = sum$; $pruned = false$;
		\For{each $s$ \ in $S$}
		\State 	$IRP = IR(s, T) * P(s)$;
		\If{$IRP  <  {\theta - 1 + P(s)}$}//Lemma 2
		\State $pruned = false$;break; 
		\EndIf
		
		\If{$s \cap T \neq \emptyset$}
		\If{$IRP  <  {\theta - (sum - P(s))}$}//Lemma 3
		\State $pruned = true$; break;
		\EndIf
		\State $remPS = remPS- P(s)$;
		\EndIf     
		\State //Lemma 4
		\If{$IRP < \theta-totalIR - remPS $}
		\State $pruned = true$; break;
		\EndIf
		\State	$totalIR = totalIR + IRP$;
		\EndFor
		\State //filtering
		\If{$pruned = false$ and $totalIR < \theta$}
		\State $result.add(T_i)$;
		\EndIf
		\EndFor
		\State 
		\Return $result$;
	\end{algorithmic}
\end{algorithm}
\subsection{Algorithm}
In Algorithm~\ref{alg:irq}, the infected rate query arguments are a query trajectory $Q$, a threshold $\theta$, and a candidate set $\mathbf{T}$, and the query result is a trajectory set of close contacts for $Q$. Initially, we let an empty ArrayList to hold the result, and $S$ is a set of segments of $Q$. Then, for each scanned trajectory $T_i$, we let $IRP = IR(s,T) * P(s)$ simplify Lemmas 2-4. The lines 5-8 for Lemma 1, lines 12-15 for Lemma 2, lines 17-19 for lemma 3, and lines 23-25 for Lemma 4. If $T_i$ does not satisfy any Lemmas 1-4, then we use the instructions (e.g., continue and break) to stop further computing $IRP$ of $T_i$ (e.g., in line 8, the sum of $P(s)$ is not greater than $\theta$, then we stop calculating $IR(Q,T_i)$ and continue to check the next trajectory $T_{i+1}$, and the $IRP$ is lower than $\theta$ in line 13, then break calculate the $IRP$ of the left segments of $Q$ and start to check next trajectory). If $T_i$ is not pruned, then it is added to the result in line 30. Then in line 33 of Algorithm~\ref{alg:irq} , we return the final result after all trajectories of $\mathbf{T}$ have been checked and calculated.

\section{Join solution}

\subsection{Basic Idea}
\begin{figure*}
	\centering
	\includegraphics[width=6.8in]{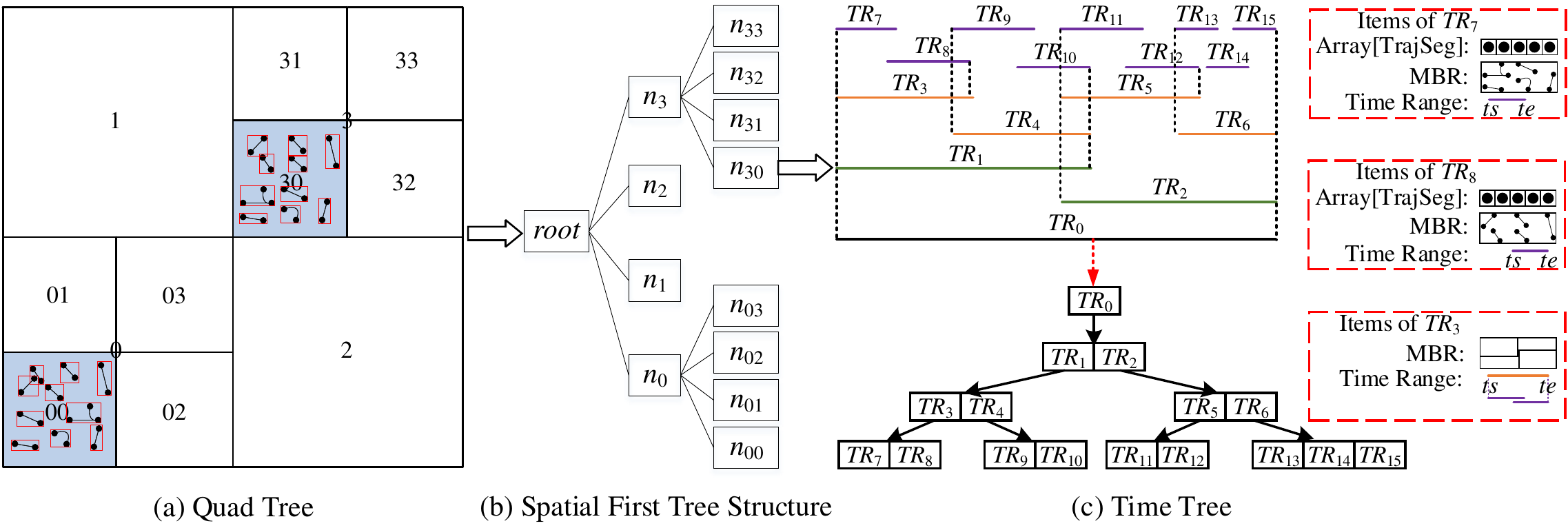}
	\caption{A sample of SFT index.}
	\label{fig:sft_index}
\end{figure*} 
To process plenty of trajectories infected rates, we develop a join query solution, named Infected Rate Join Query (IRJQ). Figure~\ref{fig:p_fm} depicts the architecture of IRJQ, which consists of three processes: Data Preprocessing, Candidates Extraction, and Infectivity Exploration.

\textbf{Data Preprocessing.} As depicted in the top-most box of Figure~\ref{fig:p_fm}, we first filter noise locations of trajectories out. After that, we extract segments of trajectories by the rules described in Figure~\ref{fig:segment}. Later, we use the XZ2T index to organize segments and then store them into the NoSQL database via JUST. The details of the steps in data preprocessing are described in Sections 3.1 and 3.2.  

\textbf{Candidates Extraction.} The middlebox of Figure~\ref{fig:p_fm} shows the procedure of candidates extraction, which allocates the coarse-grained high infect rate candidates to the query set. First, we load the query set; Second, we detect segments of this set; Third, we build an SFT-Tree (Spatial First Time Tree) index for all segments, which can minimize queries and communication on a small of data redundancy. Then, we query the candidates covered by the spatio-temporal ranges of leaf nodes in the SFT-Tree. This procedure will be a detailed introduction in Section 5.2.

\textbf{Infectivity Exploration.} The bottom-most box of Figure~\ref{fig:p_fm} shows the procedure of infectivity exploration. First, we prune some candidates, which are not the close contacts; Then, we calculate the infectivity of the segment. Meanwhile, we use a map to record some trajectory that can be further pruned. Later, we group all infectivity of the same trajectory's segments, and calculate the trajectory infectivity; Last, we filter trajectories whose infectivity is lower than $\theta$. This procedure is described in detail in Section 5.3.

\subsection{Candidates Extraction}
Candidates Extraction procedure is based on the segmentation algorithm, an SFT index strategy, and spatio-temporal query. The goal of segmentation is to diminution the long and large query trajectories on several individual segments.  We built an SFT index on all query segments, where segments with similar spatio-temporal range are placed in the same SFT leaf node. Then, we extract candidate segments from the database by the extended spatio-temporal range of each leaf node of SFT. Finally, we get coarse-grained candidates from the database with low data redundancy and a small I/O consumption.

\textbf{SFT index.} SFT (Spatial first Time index) is a two layers index. The segments of the query set are allocated to the leaf node of the SFT index with a suitable spatio-temporal region. Dividing segments into different spatio-temporal ranges can decrease I/O cost and reduce the size of the candidate set when query. As shown is Figure~\ref{fig:sft_index}, the SFT index structure is established as follows. First, we divide the spatial domain into four equal-sized regions, numbered from 0 to 3. Next, each region is recursed until reaching the maximum resolution. Then we build a time tree for each leaf node of the spatial first tree. As depicted in Figure~\ref{fig:sft_index} (c), data in each spatial region is indexed in time using a one-dimensional R-Tree-like structure~\cite{guttman1984r}. The internal node of the time tree has a one-dimensional time range and the MBR of all the leaf nodes it contains. The leaf node store the segments and the spatio-temporal range of all the segments. 

For each segment, there is data with spatial and temporal information. As shown in Figure~\ref{fig:sft_index} (a), we first allocate the segments whose lower points of their MBR are located in the same node of the quadtree to the leaf node. Then, for segments in the same spatial node, we construct the T-Tree according to the start times of their time ranges. The index of the time layer is shown in Figure~\ref{fig:sft_index} (c). In the actual construction process, in order to reduce the size of the T-Tree, we  improve its insertion and add the function of merging consecutive time ranges. Assuming that the time range $[t_1,t_2]$ has reached T in T-Tree, a new time range $[t_3,t_3]$ also reaches T, and the two time ranges cross, then we combine $[t1, t2]$ and $[t3, t4]$ together, delete the $[t1, t2]$ node and insert the $[t1, t4]$ node, which not only reduces the size of the T-Tree, but also facilitates merging time ranges. 

When the data in a node of T-Tree reaches the threshold, it will be split, e.g., $TR_7$ and $TR_8$ are crossed in time, but their merged time range is more extensive than the threshold. Thus they are regarded as two leaves, respectively. The T-Tree split will be re-divided into two parts according to the new time range, and then the T-Tree index will be rebuilt.

\textbf{ST Query.} After building the SFT index, the segments of the query set are distributed in leaf nodes of the SFT index with individual spatio-temporal range. Therefore, we start with the root node, and when accessing the leaf node, we extract its spatio-temporal range and expand this range to the infection range, and then query candidates from the database like Section 4.2.2. 
\begin{algorithm}[htb]
	\caption{ Infected Rate Join Query.}
	\label{alg:irjq}
	\begin{algorithmic}[1]
		\Require The query set $\mathbf{Q}$, A threshold $\theta$.
		\Ensure  All trajectory pairs form $\mathbf{Q}$ and database whose infectivity exceeds $\theta$ 
		\State $sft = new~SFT(); ~//new~ index$
		\For{each $Q \in \mathbf{Q}$}
		\State $S = segmentation(Q)$;
		\State $sft.insert(S)$
		\EndFor
		\State $removed = new~Map()$,$remain = new~Map()$;
		\State $result= new~ArrayList()$;
		\State search($sft.root$, $removed$,$remain$,$result$);
		\State $result.filter(removed, remain)$; //filtering.
		\State $finalResult =result.reduceByKey(v_1+v_2)$;
		\State \Return $finalResult.filter(v \ge \theta)$;
	\end{algorithmic}
\end{algorithm}

\subsection{Infectivity Exploration} 
The infected rate join query needs to return all pairs of two sets. A trajectory $Q_i$ of the query set $\mathbf{Q}$ are divided into several segments distributed in leaf nodes of the SFT index. In Section 5.2, we have extracted candidates for each leaf node. Thus, we should first calculate the segment infected rate in each node where $Q_i$ has segment located. Then, we merge all the segment infected rates between $Q_i$ and other trajectories, respectively. However, the calculate of the infected rate is time-consuming. Therefore, we use some pruning strategies, which do not need to calculate the infected rate of all candidates of $Q_i$ to exclude those trajectories that are not close to $Q_i$ in spatial and temporal. 

\textbf{Pruning.} In each leaf node with $Q_i$, we first filter segments for $Q_i$ whose spatial or temporal distance exceeds $\theta_d$ or $\theta_t$. The pruning strategy Lemma~\ref{theorem2} can also be used to the infected rate join query, but we need to use two maps to record which trajectories have been cropped and which trajectories can be used for further calculations. Then, we only need to calculate the trajectories in the remaining map for $Q_i$.

\begin{algorithm}[htb]
	\caption{Extract candidates and pruning function.}
	\label{alg:search}
	\begin{algorithmic}[1]
		\Require A node of SFT index $node$; $removed$ records the pruned trajectory pair; $remain$ records remain candidate trajectories; $result$ records the value of segments IRP.
		\Ensure
		\If{$node.type = LEAF$}
		\State $candidates = st\_query(node.mbr, node.tr, \theta_d,\theta_t)$
		\For{each $s \in node.data$}
		
		\For{each $c \in candidates$}
		\State // $st\_filter$ return true when the spatial or temporal distance exceeds $\theta_d$ or $\theta_t$, else return false
		\If{$st\_filter(s,c)$} 	continue;
		\EndIf
		\If{$remain \neq \emptyset$ and $ !remain.has(s.id, c.id)$}
		continue;
		\EndIf
		
		\If{$remove.has(s.id, c.id)$}
		continue;
		\EndIf
		
		\State $IRP = IR(s, c) * P(s);$
		\If{$IRP < \theta - 1 + P(s)$}  // Lemma 2
		\If{$remain \neq \emptyset$}
		\State $remain.remove(s.id, c.id)$;
		\State continue; 
		\EndIf
		\State $remove.put(s.id, c.id)$
		\State continue; 
		\EndIf
		\State $remain.put(s.id, c.id)$;
		\State $result.add((s.id, c.id),IRP)$;
		\EndFor
		\EndFor
		\Else{}// search the children nodes.
		\State search(node.ne);search(node.se);
		\State search(node.sw);search(node.nw);
		\EndIf	
	\end{algorithmic}
\end{algorithm}
\subsection{Algorithm}
The pseudocode of IRJQ is shown in Algorithm~\ref{alg:irjq} and  Algorithm~\ref{alg:search}. The query arguments are the query set $\mathbf{Q}$ and a threshold $\theta$ in Algorithm~\ref{alg:irjq}. We first build an SFT index $sft$ for segments of the query set $\mathbf{Q}$ in lines 1-5 of  Algorithm~\ref{alg:irjq}. Then, we search from the root node of $sft$ through depth-first strategy. The search processing is shown in Algorithm~\ref{alg:search}.  In lines 1-24 of Algorithm~\ref{alg:search}  calculate the $IR$ between segments and candidates in the leaf node. If the visited node is not the leaf node, then search the child nodes while it has, as shown in lines 26-27. In line 2, we call spatio-temporal query to extract candidates from the basic database. For each segment $s$ in the leaf node, we scan candidates one by one. If the spatio-temporal of any candidate $c$ is not covered by $s$, then skip $c$. line 12, we calculate $IRP =IR(s,c) * P(s)$. In line 13, we judge $IRP$ by lemma 2. If it lower than $\theta - 1 + P(s)$, then update $remain$ and $removed$ map and continue to check the next candidate, else we put the candidate in $remain$ map and record $IRP$ for $(s.id,c.id)$ in $result$. After visiting all leaf nodes, in line 9 of Algorithm~\ref{alg:irjq}, we refine the result, where candidates in $removed$ map or not in $remain$ map will be pruned. In lines 10-11, we reduce the result by key (key = $s.id$ + $c.id$) and then filter the final results out whose value are lower than $\theta$.

\section{EXPERIMENTS}
We have implemented our algorithms and conducted extensive experiments on real and synthetic spatial data sets to verify our proposed techniques. 
\setlength{\tabcolsep}{0.1em}
\begin{table}
	\centering
	\caption{Trajectory Data Sets}
	\begin{tabular}{|c|c|c|c|c|c|}
		\hline
		Attributes & MPG & MPG2 & MPG3 & MPG4 & MPG5 \\ \hline
		\# Points & 3,079,428 & 16,397,140 & 19,476,568 & ... & 25,635,424 \\ \hline
		\# Traj. & 160,840 & 33,680 &  194,520 & ... & 516,200 \\ \hline
	\end{tabular}
	\label{datasets}
\end{table}
\setlength{\tabcolsep}{0.7em}
\begin{table}
	\centering
	\caption{Default Parameters}
	\begin{tabular}{|c|c|c|c|c|c|} \hline
		$\lambda$ &	$\theta$ &$\theta_d $ &$\theta_t$ &Resolution & Query Size \\ \hline
		0.5 & 0.5 & 50m & 120s & 15 & 2,300 Traj.\\
		\hline\end{tabular}
	\label{def_param}
\end{table}
\begin{table}
	\centering
	\caption{Pruning Effectiveness \hspace{\fill} (Unit: ms)}
	\begin{tabular}{|c|c|c|c|c|} \hline
		&IRQ &	IRQ\_UP&IRJQ &IRJQ\_UP \\ \hline
		MPG & \textbf{701} & 963 & \textbf{215} & 260  \\ \hline
		MPG5 & \textbf{2647} & 3182 & \textbf{318} & 394  \\
		\hline\end{tabular}
	\label{pruning_eff}
\end{table}

\begin{figure*}[htbp]
	\centering

	\subfigure[Runtime of MPG.]{
		\begin{minipage}[t]{0.25\linewidth}
			\centering
			\includegraphics[width=1\textwidth]{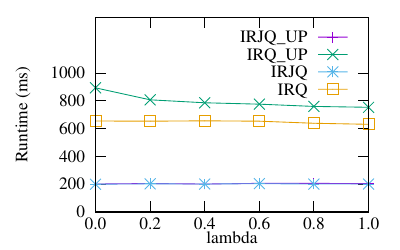}
		\end{minipage}%
	}%
	\subfigure[Runtime of MPG5.]{
		\begin{minipage}[t]{0.25\linewidth}
			\centering
			\includegraphics[width=1\textwidth]{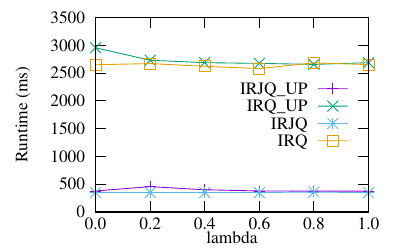}
		\end{minipage}%
	}%
	\subfigure[Accuracy of MPG.]{
		\begin{minipage}[t]{0.25\linewidth}
			\centering
			\includegraphics[width=1\textwidth]{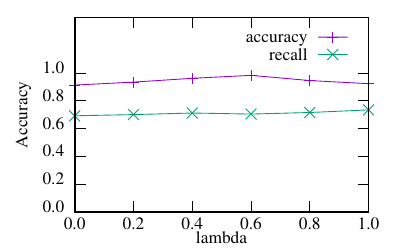}
		\end{minipage}%
	}%
	\subfigure[Accuracy of MPG5.]{
		\begin{minipage}[t]{0.25\linewidth}
			\centering
			\includegraphics[width=1\textwidth]{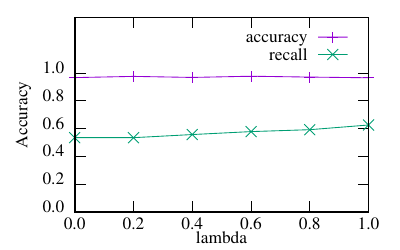}
		\end{minipage}%
	}%
	\caption{The effect of $\lambda$.}
	\label{fig:lambda}
\end{figure*}

\begin{figure*}[htbp]
	\centering
	\subfigure[Runtime of MPG.]{
		\begin{minipage}[t]{0.25\linewidth}
			\centering
			\includegraphics[width=1\textwidth]{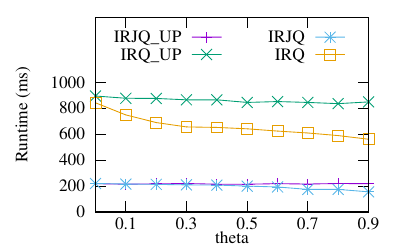}
		\end{minipage}%
	}%
	\subfigure[Runtime of MPG5.]{
		\begin{minipage}[t]{0.25\linewidth}
			\centering
			\includegraphics[width=1\textwidth]{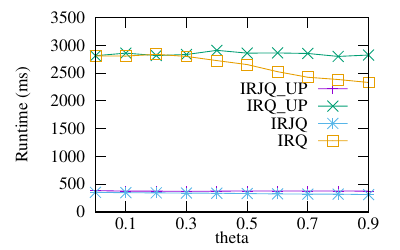}
		\end{minipage}%
	}%
	\subfigure[Accuracy of MPG.]{
		\begin{minipage}[t]{0.25\linewidth}
			\centering
			\includegraphics[width=1\textwidth]{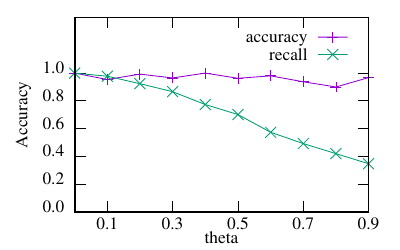}
		\end{minipage}%
	}%
	\subfigure[Accuracy of MPG5.]{
		\begin{minipage}[t]{0.25\linewidth}
			\centering
			\includegraphics[width=1\textwidth]{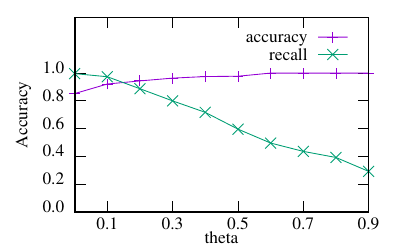}
		\end{minipage}%
	}%
	\caption{The effect of $\theta$.}
	\label{fig:theta}
\end{figure*}
\textbf{Datasets.} To evaluate the efficiency and correctness, we use the GPS records from the mobile phones, \textbf{MPG}~\footnote{http://suo.im/69LJCp} and Synthetic data sets (MPG2, MPG3, MPG4, MPG5) which are generated by copying \& offsetting one to four times of \textbf{MPG} to test the scalability of our solutions. As shown in Table~\ref{datasets}, there are 160,840 trajectories in \textbf{MPG}, with an average of 19 points. The query trajectory set$^1$ is a labeled down-sampled data set of the \textbf{MPG}.

\textbf{Setting.} All of the algorithms were implemented in Java and Scala. All the experiments were conducted on a cluster of 5 nodes, with each node equipped with CentOS 7.4, 8-core CPU, 32GB RAM, and 1T disk. In our experiments, we compare the run time and veracity of IRQ and IRJQ. The run time is the average query time. The accuracy is the portion of correctly labeled trajectories in the query result, and the recall is the number of query trajectories whose query results are not empty. We analyze the effect of preference parameter $\lambda$, precision threshold $\theta$, distance threshold $\theta_d$, time threshold $\theta_t$, query data size, and source data size for IRQ and IRJQ. We also verify the effect of the different resolution of quadtree on the IRJQ algorithm. Table~\ref{def_param} gives the default parameters. 

\subsection{Pruning Effectiveness}
We first analyze the pruning effectiveness of our algorithms using the default parameters. The experimental results are shown in Table~\ref{pruning_eff}.  We can see that the IRQ and IRJQ have better performance than the unpruned algorithms IRQ\_UP and IRJQ\_UP. The join algorithms (e.g., IRJQ and IRJQ\_UP) outperform simple query algorithms (IRJQ and IRJQ\_UP) by almost an order of magnitude on MPG and MPG5. After pruning, IRQ saves 27\% and 17\%, and IRJQ improves 27\% and 29\% query time on the sets of MPG, respectively.
\begin{figure*}[htbp]
	\centering

	\subfigure[Runtime of MPG.]{
		\begin{minipage}[t]{0.25\linewidth}
			\centering
			\includegraphics[width=1\textwidth]{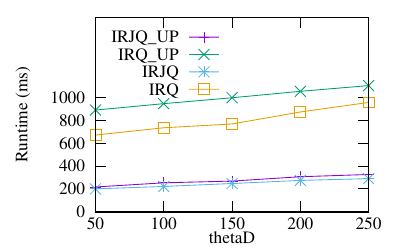}
		\end{minipage}%
	}%
	\subfigure[Runtime of MPG5.]{
		\begin{minipage}[t]{0.25\linewidth}
			\centering
			\includegraphics[width=1\textwidth]{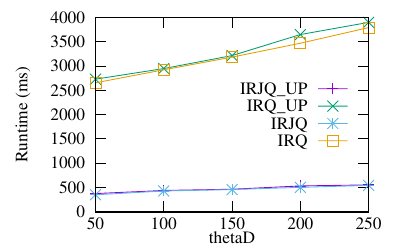}
		\end{minipage}%
	}%
	\subfigure[Accuracy of MPG.]{
		\begin{minipage}[t]{0.25\linewidth}
			\centering
			\includegraphics[width=1\textwidth]{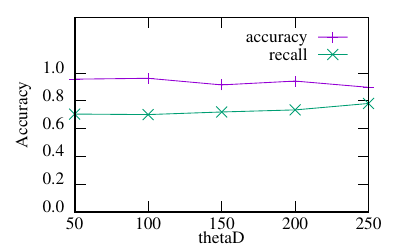}
		\end{minipage}%
	}%
	\subfigure[Accuracy of MPG5.]{
		\begin{minipage}[t]{0.25\linewidth}
			\centering
			\includegraphics[width=1\textwidth]{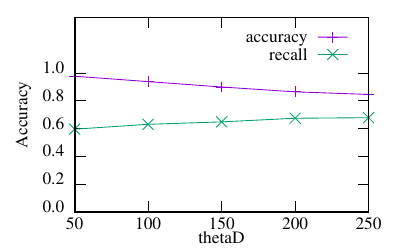}
		\end{minipage}%
	}%
	\caption{The effect of $\theta_d$.}
	\label{fig:thetaD}
\end{figure*}
\begin{figure*}[htbp]
	\centering

	\subfigure[Runtime of MPG.]{
		\begin{minipage}[t]{0.25\linewidth}
			\centering
			\includegraphics[width=1\textwidth]{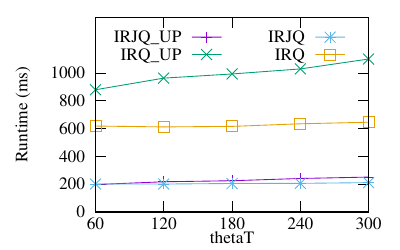}
		\end{minipage}%
	}%
	\subfigure[Runtime of MPG5.]{
		\begin{minipage}[t]{0.25\linewidth}
			\centering
			\includegraphics[width=1\textwidth]{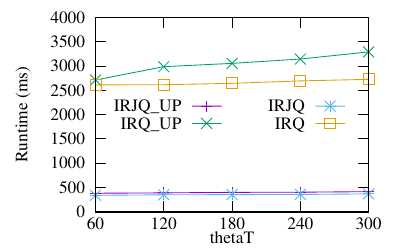}
		\end{minipage}%
	}%
	\subfigure[Accuracy of MPG.]{
		\begin{minipage}[t]{0.25\linewidth}
			\centering
			\includegraphics[width=1\textwidth]{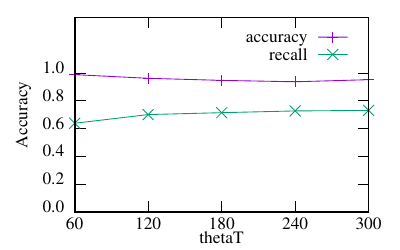}
		\end{minipage}%
	}%
	\subfigure[Accuracy of MPG5.]{
		\begin{minipage}[t]{0.25\linewidth}
			\centering
			\includegraphics[width=1\textwidth]{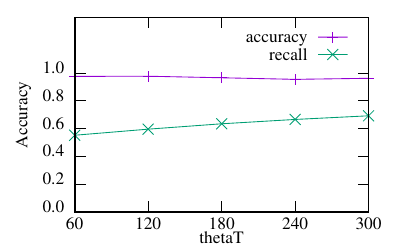}
		\end{minipage}%
	}%
	\caption{The effect of $\theta_t$.}
	\label{fig:thetaT}
\end{figure*}
\subsection{Effect of Preference Parameter $\lambda$}
Figure~\ref{fig:lambda} shows the effect of the preference parameter $\lambda$ on efficiency and accuracy. On the dataset MPG, the runtime of IRJQ and IRJQ\_UP is around 200ms and not affected by the varying of $\lambda$. IRQ is more effective than IRQ\_UP, and saves about 200ms compare to IRQ\_UP. On the dataset MPG5, the runtime of IRJQ is around 400ms and much faster than IRQ. Then we discuss the recall and accuracy under the influence of $\lambda$. We only analyze the recall and accuracy of IRQ because IRQ and IRJQ have the same accuracy. The accuracy of IRQ varies with $\lambda$, which is lower at $\lambda=0.0$ and $\lambda=1.0$ on dataset MPG because only the side affected by spatial or temporal is considered. As the number of candidates from the spatio-temporal query on the database does not change but the final result has a little difference. Thus the recall rate remains at around 0.7 on the dataset MPG. However, on the dataset MPG5, the recall increases with $\lambda$, because more trajectories are queried in the same spatial region, which increases the number of the final results.
\subsection{Effect of Threshold $\theta$}
Figure~\ref{fig:theta} shows the effect of precision threshold $\theta$. On the datasets of MPG and MPG5, the runtime of pruned algorithms (IRQ and IRJQ) decreases with an increasing threshold $\theta$ and unpruned algorithms (IRQ\_UP and IRJQ\_UP) keep the approximately same value, which verify the efficiency of our pruning strategies. The recall decreases with the threshold because the larger $\theta$, the fewer trajectories are satisfied. Meanwhile, the accuracy rate still maintains at a relatively high value. 

\subsection{Effect of Distance Threshold $\theta_d$}
Figure~\ref{fig:thetaD} shows the effect of distance threshold $\theta_d$. The running time increases as  $\theta_d$ increases because both the candidates and the computation are increased. IRJQ's runtime only increases slightly, but the increase in IRQ is pronounced. The spatio-temporal ranges infected by every location are enlarged with the increase of $\theta_d$. Thus more candidates are selected in our algorithms, which increases the recall.

\subsection{Effect of Time Threshold $\theta_t$}
Figure~\ref{fig:thetaT} shows the effect of distance threshold $\theta_t$. The running time slightly increases as $\theta_t$ increases. The temporal range expands as $\theta_t$ increases, but the data covered by each location does not increase significantly. Therefore, the running time maintains at a relatively fixed value, and the recall increases with a small numerical.

\subsection{Effect of Query Size}
Figure~\ref{fig:query_size} shows the effect of query size. The query trajectories from 460 to 2300, represented as 20\% to 100\%, respectively. Note that unlike other times, the runtime here is the total time to complete the calculation of all trajectories. We can see that as the amount of data for query trajectories doubles, the cumulative runtime of IRQ and IRQ\_UP almost doubles. Consequently, the total runtime of IRJQ and IRJQ\_UP is still below 500 seconds, because we build an efficient index on trajectories that effectively reduce the number of time to access the database and greatly avoid data redundancy.
\begin{figure}
	\centering
	\includegraphics[width=2.0in]{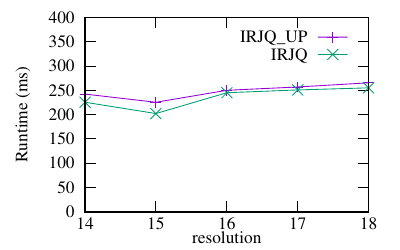}
	\caption{The effect of resolution.}
	\label{fig:resolution}
\end{figure}

\subsection{Effect of Resolution}
In Figure~\ref{fig:resolution}, we analyze the impact of the resolution of the SFT index. We see that the resolution equals 15 performs better than other resolutions. When the resolution is equal to 14, the spatial-temporal range distribution in the leaf nodes is more dispersed. Thus the spatial and temporal range of a single query will be relatively large, which will cause more data redundancy and some unnecessary data. When the resolution is greater than 15, although the distribution of the spatio-temporal range in the leaf node is very concentrated, it also means that it needs to query the database more times, which increases the I/O overhead.

\subsection{Scalability}
Figure~\ref{fig:data_size} shows the effect of basic dataset size. We generate five datasets with equal increments in turn, as shown in Table~\ref{datasets}. We represent MGP as 20\%, MPG2 as 40\%, MPG3 as 60\%, MPG2 as 80\%, and MPG5 as 100\%, respectively. We see that the increase of the data set has a great impact on the IRQ, which owns to spatio-temporal query gets several times of the candidate set, which increases the calculations. Although IRJQ's query time has also increased, it has not increased exponentially. It is attributed to the SFT index, which makes the close spatio-temporal ranges query only once on the database.

\begin{figure*}[htbp]
	\centering 
	\begin{minipage}[b]{0.48\textwidth} 
		\centering

		\subfigure[Runtime of MPG.]{
			\begin{minipage}[t]{0.5\linewidth}
				\centering
				\includegraphics[width=1\textwidth]{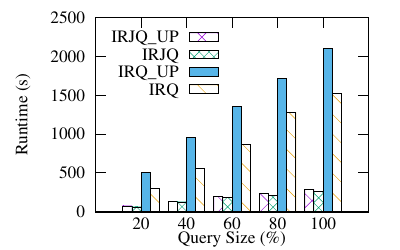}
			\end{minipage}%
		}%
		\subfigure[Accuracy of MPG5.]{
			\begin{minipage}[t]{0.5\linewidth}
				\centering
				\includegraphics[width=1\textwidth]{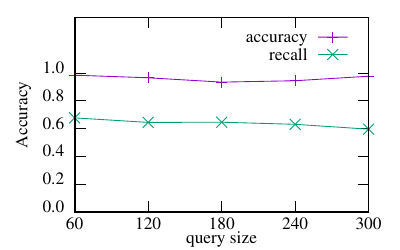}
			\end{minipage}%
		}%
		\caption{The effect of query size.}
		\label{fig:query_size}
	\end{minipage}
	\begin{minipage}[b]{0.48\textwidth} 
		\centering

		\subfigure[Runtime of MPG.]{
			\begin{minipage}[t]{0.5\linewidth}
				\centering
				\includegraphics[width=1\textwidth]{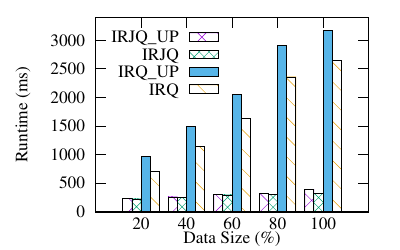}
			\end{minipage}%
		}%
		\subfigure[Accuracy of MPG5.]{
			\begin{minipage}[t]{0.5\linewidth}
				\centering
				\includegraphics[width=1\textwidth]{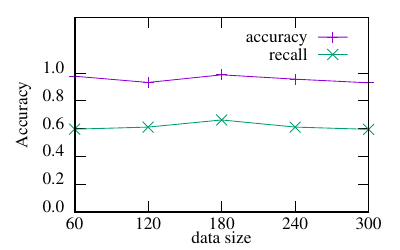}
			\end{minipage}%
		}%
		\caption{The effect of data size.}
		\label{fig:data_size}
	\end{minipage}
\end{figure*}

\section{Related Work}
\subsection{Trajectory Correlation Metrics}
Many trajectory related metrics have been proposed~\cite{toohey2015trajectory,assent2009anticipatory,deng2014massjoin,yu2016string,shang2017trajectory,ding2008efficient,lin2005shapes,zheng2013towards,rong2017fast,wang2013effectiveness,magdy2015review}, which can be roughly classified as two types: (1) The point-based metrics, such as the Euclidean distance (ED)~\cite{faloutsos1994fast}, Dynamic Time Warping (DTW)~\cite{assent2009anticipatory} and  Fréchet~\cite{toohey2015trajectory}; and (2) The segment based metrics, such as the metric in ~\cite{xie2017distributed} and Longest Common Subsequence (LCSS)~\cite{zhu2017trajectory}. In general, the above methods all treat temporal attributes as simple time series. The temporal attribute is one of the important attributes. Combining the spatial and temporal attributes can be used as a judgment criterion for the correlation measure of trajectory.
\\ \textbf{Point-Based Metrics.} Correlation measurement methods based on trajectory points can be further divided into global matching methods and local matching methods. Euclidean distance~\cite{faloutsos1994fast} is mainly by calculating the Euclidean distance between corresponding points between trajectories, and then accumulating the sum as the final metric value. Euclidean distance is the simplest. It only needs to be summed up, and the time complexity is O(N). However, it also has serious limitations: (1) The sampling rate and trajectory points must be consistent; (2) The principle of monotonic continuity must be met, and local time distortion is not supported; (3) Sensitive to noise. The Fréchet distance~\cite{toohey2015trajectory} measure was proposed by Fréchet et al. It is usually described intuitively: the dog rope distance when a person walks a dog. DTW~\cite{assent2009anticipatory} locally stretches or scales the trajectories, so that trajectories of different sampling rates and different lengths can be compared. The DTW distance is the cumulative sum of the distances between all the optimal matching trajectory points. 
\\ \textbf{Segment-Based Metrics.} Based on the trajectory segment similarity measures, by segmenting the trajectories and comparing the similarity of each segment separately, the time complexity is greatly reduced. However, the local information of the trajectory is not fully considered. Thus the accuracy is relatively low. Longest Common Subsequence (LCSS) mainly considers similar parts between trajectories as a measure of trajectory correlation, so it can skip some trajectory points due to matching distance exceeding the threshold, which makes it robust to noise. In~\cite{xie2017distributed}, the trajectory is divided into several segments, and then the divide-and-conquer strategy is used to calculate the discrete segment Hausdorff or discrete segment Fréchet distance. However, the distance measures they use do not adequately describe spatial and temporal proximity.

Our metric takes into account the spatial and temporal correlation of each location at the same time and uses a divide-and-conquer algorithm for the weighted trajectory segments. It is not only appropriately describes the spatio-temporal closed trajectories, but also greatly reduces the complexity.
\subsection{Trajectory Correlation Search.}
Trajectory correlation searches are widely studied~\cite{shang2017trajectory,shang2018parallel,ta2017signature,xie2017distributed,rong2017fast}. The procedure typically involves a definition step and a query processing step. First, a metric is defined to measure the spatial and temporal correlations between two trajectories. Second, an efficient strategy is developed to search spatiotemporally close trajectories for a query trajectory. For example, the BCT~\cite{zheng2017searching} algorithm proposed by Zheng uses Euclidean distance for trajectory search. Rong et al.~\cite{rong2017fast} proposed a similar measure for trajectory segments and used a distributed framework. This framework first divides trajectories into several segments and then groups nearby segments to find common trajectories, which is helpful to reduce I/O consumption. Shang et al.~\cite{shang2017trajectory} presented a two-phase divide-and-conquer trajectory similarity join framework. It first finds similar trajectories for each trajectory. Then it merges the results to the final result. Like~\cite{rong2017fast}, we proposed a clearer metric, and then grouped segments that are closed spatially and temporally with an efficient Spatial First Time (SFT) index.
Furthermore, in order to speed up query time, many frameworks have designed effective pruning strategies, which reduce the search space. Many studies~\cite{shang2017trajectory,shang2018parallel} have analyzed the low bounds of their metrics to reduce search nodes. In our algorithm, we propose four effective prune strategies to avoid unnecessary calculations.
	
\section{CONCLUSION}
In this paper, we studied a novel trajectory infection rate based on spatio-temporal correlation, the goal of which is to detect suspected infected crowds of COVID-19 and also targets many applications such as close contacts detection, companion mining, high-risk groups prediction and trajectory similarity recommendation. We proposed a new trajectory metric that accommodates misaligned trajectory points. We first broke down the longer and larger trajectories into several short and suitable segments and used an active spatio-temporal index (XZ2T) to manage a large number of segments in the NoSQL database, which reduce memory consumption and guarantee the scalability by avoid loading all trajectories into memory. We then developed efficient algorithms for segments with infected weight. We explored several pruning strategies for our proposed algorithms to avoid many calculations. For batch query, we designed an SFT index that groups similar segments only once to access the database to reduce I/O communication and data redundancy. We then devised experimental studies on real and synthetic datasets to verify the effectiveness and efficiency of our algorithms.

Many exciting directions for future research exist. First, it is significant to extend our algorithms for supporting top-$k$ close contacts query without a threshold $\theta$. Second, it is vital  to use some sample data to determine the parameters to be set in the algorithm.
	
\section{ACKNOWLEDGEMENT}
This work was supported by the National Key R\&D Program of China (2019YFB2101800).
	
	
	\balance
	

	\bibliographystyle{abbrv}
	\bibliography{main} 

\begin{thebibliography}{10}

\bibitem{assent2009anticipatory}
I.~Assent, M.~Wichterich, R.~Krieger, H.~Kremer, and T.~Seidl.
\newblock Anticipatory dtw for efficient similarity search in time series
  databases.
\newblock {\em Proceedings of the VLDB Endowment}, 2(1):826--837, 2009.

\bibitem{bakalov2005efficient}
P.~Bakalov, M.~Hadjieleftheriou, E.~Keogh, and V.~J. Tsotras.
\newblock Efficient trajectory joins using symbolic representations.
\newblock In {\em Proceedings of the 6th international conference on Mobile
  data management}, pages 86--93, 2005.

\bibitem{bakalov2006continuous}
P.~Bakalov and V.~J. Tsotras.
\newblock Continuous spatiotemporal trajectory joins.
\newblock In {\em International conference on GeoSensor Networks}, pages
  109--128. Springer, 2006.

\bibitem{boxhm1999xz}
C.~B{\"O}xhm, G.~Klump, and H.-P. Kriegel.
\newblock Xz-ordering: A space-filling curve for objects with spatial
  extension.
\newblock In {\em International Symposium on Spatial Databases}, pages 75--90.
  Springer, 1999.

\bibitem{deng2014massjoin}
D.~Deng, G.~Li, S.~Hao, J.~Wang, and J.~Feng.
\newblock Massjoin: A mapreduce-based method for scalable string similarity
  joins.
\newblock In {\em 2014 IEEE 30th International Conference on Data Engineering},
  pages 340--351. IEEE, 2014.

\bibitem{ding2008efficient}
H.~Ding, G.~Trajcevski, and P.~Scheuermann.
\newblock Efficient similarity join of large sets of moving object
  trajectories.
\newblock In {\em 2008 15th International Symposium on Temporal Representation
  and Reasoning}, pages 79--87. IEEE, 2008.

\bibitem{faloutsos1994fast}
C.~Faloutsos, M.~Ranganathan, and Y.~Manolopoulos.
\newblock Fast subsequence matching in time-series databases.
\newblock {\em Acm Sigmod Record}, 23(2):419--429, 1994.

\bibitem{guttman1984r}
A.~Guttman.
\newblock R-trees: A dynamic index structure for spatial searching.
\newblock In {\em Proceedings of the 1984 ACM SIGMOD international conference
  on Management of data}, pages 47--57, 1984.

\bibitem{li2020just}
R.~Li, H.~He, R.~Wang, Y.~Huang, J.~Liu, S.~Ruan, T.~He, J.~Bao, and Y.~Zheng.
\newblock Just: Jd urban spatio-temporal data engine.
\newblock {\em ICDE. IEEE}, 2020.

\bibitem{lin2005shapes}
B.~Lin and J.~Su.
\newblock Shapes based trajectory queries for moving objects.
\newblock In {\em Proceedings of the 13th annual ACM international workshop on
  Geographic information systems}, pages 21--30, 2005.

\bibitem{magdy2015review}
N.~Magdy, M.~A. Sakr, T.~Mostafa, and K.~El-Bahnasy.
\newblock Review on trajectory similarity measures.
\newblock In {\em 2015 IEEE seventh international conference on Intelligent
  Computing and Information Systems (ICICIS)}, pages 613--619. IEEE, 2015.

\bibitem{rong2017fast}
C.~Rong, C.~Lin, Y.~N. Silva, J.~Wang, W.~Lu, and X.~Du.
\newblock Fast and scalable distributed set similarity joins for big data
  analytics.
\newblock In {\em 2017 IEEE 33rd International Conference on Data Engineering
  (ICDE)}, pages 1059--1070. IEEE, 2017.

\bibitem{cloudtp}
S.~{Ruan}, R.~{Li}, J.~{Bao}, T.~{He}, and Y.~{Zheng}.
\newblock Cloudtp: A cloud-based flexible trajectory preprocessing framework.
\newblock In {\em 2018 IEEE 34th International Conference on Data Engineering
  (ICDE)}, pages 1601--1604, April 2018.

\bibitem{shang2017trajectory}
S.~Shang, L.~Chen, Z.~Wei, C.~S. Jensen, K.~Zheng, and P.~Kalnis.
\newblock Trajectory similarity join in spatial networks.
\newblock 2017.

\bibitem{shang2018parallel}
S.~Shang, L.~Chen, K.~Zheng, C.~S. Jensen, Z.~Wei, and P.~Kalnis.
\newblock Parallel trajectory-to-location join.
\newblock {\em IEEE Transactions on Knowledge and Data Engineering},
  31(6):1194--1207, 2018.

\bibitem{shang2014personalized}
S.~Shang, R.~Ding, K.~Zheng, C.~S. Jensen, P.~Kalnis, and X.~Zhou.
\newblock Personalized trajectory matching in spatial networks.
\newblock {\em The VLDB Journal}, 23(3):449--468, 2014.

\bibitem{shang2014discovery}
S.~Shang, K.~Zheng, C.~S. Jensen, B.~Yang, P.~Kalnis, G.~Li, and J.-R. Wen.
\newblock Discovery of path nearby clusters in spatial networks.
\newblock {\em IEEE Transactions on Knowledge and Data Engineering},
  27(6):1505--1518, 2014.

\bibitem{ta2017signature}
N.~Ta, G.~Li, Y.~Xie, C.~Li, S.~Hao, and J.~Feng.
\newblock Signature-based trajectory similarity join.
\newblock {\em IEEE Transactions on Knowledge and Data Engineering},
  29(4):870--883, 2017.

\bibitem{tang2012discovery}
L.-A. Tang, Y.~Zheng, J.~Yuan, J.~Han, A.~Leung, C.-C. Hung, and W.-C. Peng.
\newblock On discovery of traveling companions from streaming trajectories.
\newblock In {\em 2012 IEEE 28th International Conference on Data Engineering},
  pages 186--197. IEEE, 2012.

\bibitem{toohey2015trajectory}
K.~Toohey and M.~Duckham.
\newblock Trajectory similarity measures.
\newblock {\em Sigspatial Special}, 7(1):43--50, 2015.

\bibitem{vinten2003cholera}
P.~Vinten-Johansen, H.~Brody, N.~Paneth, S.~Rachman, D.~Zuck, M.~Rip, H.~C.
  A.~D. Zuck, et~al.
\newblock {\em Cholera, chloroform, and the science of medicine: a life of John
  Snow}.
\newblock Medicine, 2003.

\bibitem{wang2013effectiveness}
H.~Wang, H.~Su, K.~Zheng, S.~Sadiq, and X.~Zhou.
\newblock An effectiveness study on trajectory similarity measures.
\newblock In {\em Proceedings of the Twenty-Fourth Australasian Database
  Conference-Volume 137}, pages 13--22. Australian Computer Society, Inc.,
  2013.

\bibitem{xie2017distributed}
D.~Xie, F.~Li, and J.~M. Phillips.
\newblock Distributed trajectory similarity search.
\newblock {\em Proceedings of the VLDB Endowment}, 10(11):1478--1489, 2017.

\bibitem{yu2016string}
M.~Yu, G.~Li, D.~Deng, and J.~Feng.
\newblock String similarity search and join: a survey.
\newblock {\em Frontiers of Computer Science}, 10(3):399--417, 2016.

\bibitem{zhao2020comparative}
D.~Zhao, F.~Yao, L.~Wang, L.~Zheng, Y.~Gao, J.~Ye, F.~Guo, H.~Zhao, and R.~Gao.
\newblock A comparative study on the clinical features of covid-19 pneumonia to
  other pneumonias.
\newblock {\em Clinical Infectious Diseases}, 2020.

\bibitem{zheng2013towards}
K.~Zheng, S.~Shang, N.~J. Yuan, and Y.~Yang.
\newblock Towards efficient search for activity trajectories.
\newblock In {\em 2013 IEEE 29Th international conference on data engineering
  (ICDE)}, pages 230--241. IEEE, 2013.

\bibitem{zheng2017searching}
Y.~Zheng, Z.~Chen, and X.~Xie.
\newblock Searching similar trajectories by locations, Mar.~14 2017.
\newblock US Patent 9,593,957.

\bibitem{zhu2017trajectory}
L.~Zhu, J.~R. Holden, and J.~D. Gonder.
\newblock Trajectory segmentation map-matching approach for large-scale,
  high-resolution gps data.
\newblock {\em Transportation Research Record}, 2645(1):67--75, 2017.

\end{thebibliography}

\end{document}